\documentclass{article}
\usepackage{ijcai22}

\usepackage{times}

\usepackage{soul}
\usepackage{url}
\usepackage[hidelinks]{hyperref}
\usepackage[utf8]{inputenc}
\usepackage[small]{caption}
\usepackage{graphicx}
\usepackage{amsmath}
\usepackage{booktabs}
\newif\iflong
\newif\ifshort

\longtrue

\iflong
\else
\shorttrue
\fi

\usepackage{complexity}
\usepackage{amsmath}
\usepackage{amssymb}
\usepackage{amsthm}
\usepackage{xspace}
\usepackage{bm}
\usepackage[textwidth=16mm,textsize=footnotesize]{todonotes}
\usepackage{tikz}
\usetikzlibrary{positioning,calc}
\usepackage{subfigure}

\newcommand{\bigoh}{\mathcal{O}}    
\newcommand{\tempAsgn}{\ensuremath{\pi_{\mathrm{SOL}}}}
\newcommand{\problem}{envy-free graph cutting}
\newcommand{\problemshort}{\textsc{EF-GC}}
\newcommand{\contigcake}{\textsc{Contiguous Cake Cutting}}
\newcommand{\EF}{\textsc{EF}}
\newcommand{\NPP}{\textsc{Number Partitioning}}

\newcommand{\problemvd}{envy-free vertex-disjoint graph cutting}
\newcommand{\problemvdshort}{\textsc{EF-VDGC}}

\newcommand{\val}{\ensuremath{u}}

\newcommand{\III}{\mathcal{I}}

\renewcommand{\epsilon}{\varepsilon}

\DeclareMathOperator{\sign}{sign}

\DeclareMathOperator{\agentPieces}{Pieces}

\newtheorem{theorem}{Theorem}

\newtheorem{observation}[theorem]{Observation}
\newtheorem{lemma}[theorem]{Lemma}
\newtheorem{proposition}[theorem]{Proposition}

\title{The Complexity of Envy-Free Graph Cutting}

\author{
  Argyrios Deligkas$^1$
  \and
  Eduard Eiben$^1$
  \and
  Robert Ganian$^2$
  \and
  Thekla Hamm$^3$
  \and
  Sebastian Ordyniak$^4$
  \affiliations
  $^1$Royal Holloway, University of London, UK\\
  $^2$TU Wien, Austria\\
  $^3$Eindhoven University of Technology, the Netherlands\\
  $^4$University of Leeds, UK\\
  \emails
  \{argyrios.deligkas,eduard.eiben\}@rhul.ac.uk, \{rganian, thekla.hamm,sordyniak\}@gmail.com
}

\begin{document}

\maketitle

\begin{abstract}
We consider the problem of fairly dividing a set of heterogeneous divisible resources among agents with different preferences. 
We focus on the setting where the resources correspond to the edges of a connected graph, every agent must be assigned a connected piece of this graph, and the fairness notion considered is the classical envy freeness.
The problem is NP-complete, and we analyze its complexity with respect to two natural complexity measures: the number of agents and the number of edges in the graph. While the problem remains NP-hard even for instances with 2 agents, we provide a dichotomy characterizing the complexity of the problem when the number of agents is constant based on structural properties of the graph. For the latter case, we design a polynomial-time algorithm when the graph has a constant number of edges.
\end{abstract}

\section{Introduction}
Cake cutting is, without a doubt, among the most influential problems in social choice and has received significant attention in computer science, mathematics, and economics~\cite{brams1996fair,robertson1998cake,moulin2004fair,procaccia2013cake}. The cake corresponds to a heterogeneous divisible resource that is to be divided between a set of $n$ agents with different preferences in a ``fair'' manner. In this paper, the fairness concept we focus on is {\em envy-freeness}, where every agent prefers the piece of the cake they get allocated over the piece any other agent gets.

In the classical formulation of the problem, the cake is represented as an interval and the preference of every agent over the cake is given via a {\em valuation} over any subinterval of the cake. In this setting, \citeauthor{dubins1961cut}~(\citeyear{dubins1961cut}) showed that an envy-free allocation {\em always} exists for arbitrary valuations for the agents, where the piece an agent gets consists of a countable number of subintervals. \citeauthor{stromquist1980cut}~(\citeyear{stromquist1980cut}) strengthened this result by removing the possibility of pieces that consist of a ``union of crumbs''. In fact, he showed that there is an envy-free allocation where the piece of every agent is {\em contiguous}, i.e.,  it is a {\em single} subinterval.

Recently, \citeauthor{bei2021dividing}~(\citeyear{bei2021dividing}) considered a generalized version cake cutting on {\em graphs}. This augmented model allows to capture more general scenarios which cannot be represented by splitting an interval into connected pieces---consider, e.g., the task of splitting road or railway networks between companies. We note that these settings do not always give rise to large graphs: for instance, the ICE train network in Germany can be modeled as a graph with merely 23 edges. Or, for yet another example, suppose that one wants to schedule time on a high-performance computing cluster between teams (agents). Suppose
furthermore that the day is partitioned into, e.g., four
time-slots, with each time-slot being less or more desirable for
different agents. This setting, too, can easily be modeled using a graph with four edges.
In these as well as other examples, it is still desirable to ensure that each agent receives a connected piece of the graph, but the natural model for the cake is a graph where each individual edge may be split and behaves as a single, uniform piece.

Observe that depending on the setting, it may either be the case that a vertex is allocated to only a single agent (as in the case of junctions in the division of road networks over maintenance companies), or that a vertex merely acts as a bridge between edges and may be used by multiple agents (as in the case of train stations in division of railway networks over rail companies). We call the former setting ``{\em graph cutting}'' and the latter ``{\em vertex-disjoint graph cutting}''.

While both fair division problems always admit a contiguous envy free solution when the underlying graph is a path (since they are special cases of the classical cake cutting problem), this is no longer true for more general graphs. In fact, it is not hard to construct graphs with no envy-free solutions; this holds even for stars with three leafs and two agents  with identical valuations. Hence, in the setting studied here, the natural task is to {\em decide if a solution exists, and if it does to efficiently compute one}.

\smallskip
\noindent \textbf{Our Contributions.}\quad
In this work, we explore the frontiers of tractability for both variants of graph cutting under the envy free solution concept; we refer to these problems as \problemshort\ and \problemvdshort, respectively. While it is not difficult to show that both problems are \NP-complete in general, here we analyze the problem with respect to two of the most natural complexity measures that characterize the input: the number of agents and the number of edges. 

When considering the number of agents (i.e., $n$), we begin by extending the \NP-completeness lower bounds for both variants to the case of $n=2$ even on very simple graphs---as simple as two vertices plus a matching (Theorem~\ref{the:matchNPhard}). 
However, here we can show that the two variants do sometimes behave differently: while \problemshort\ is also \NP-hard on stars when $n=2$ (Theorem~\ref{the:NP-ps}), we design a polynomial-time algorithm for \problemvdshort\ on trees when $n$ is upper-bounded by an arbitrary but fixed constant (Theorem~\ref{thm:No_agents}).
In order to achieve tractability for \problemshort, we need to restrict ourselves to instances with a constant number of agents and where the graph is a tree with a constant maximum degree (Theorem~\ref{thm:treedegreeP}).

\begin{table}[ht]
  \centering
  \begin{tabular}{cccc}    \toprule
    tw & $\Delta$ & \problemshort{} & \problemvdshort{} \\ \midrule
    2 & 3 & \NP-c (Th.~\ref{thm:twNPh}) & \NP-c (Th.~\ref{thm:twNPh}) \\
    2 & 2 & \P\ (Th.~\ref{thm:cycles}) & \P\ (Th.~\ref{thm:cycles}) \\
    1 & arbitrary & \NP-c (Th.~\ref{the:NP-ps}) & \P\ (Th.~\ref{thm:No_agents}) \\
    1 & constant & \P (Th.~\ref{thm:treedegreeP}) & \P\ (Th.~\ref{thm:No_agents}) \\
\end{tabular}
\caption{Complexity of \problemshort{} and \problemvdshort{} for a constant number of agents
  for different restrictions on the treewidth (tw) and the maximum degree ($\Delta$). All \NP{}-completeness (\NP-c) results hold already for only 2 agents.}
\label{tab:comptwde}
\end{table}

In fact, we prove this is the best one can do from this perspective. Both problems become \NP-hard on instances with two agents and graphs of maximum degree 3 which are ``almost trees''---in particular, have treewidth 2 (Theorem~\ref{thm:twNPh}). On the other hand, we show that both problems are polynomial-time solvable on cycles, which are graphs of maximum degree and treewidth $2$ (Theorem~\ref{thm:cycles}); this provides a complete dichotomy for the complexity of both problems with respect to treewidth and maximum degree (see Table~\ref{tab:comptwde}).

Next, we target instances where the number of edges is bounded by a constant. As the main technical contribution of this article, we show that both problems under consideration become polynomial-time tractable under this restriction (Theorem~\ref{thm:No_edges_main}). The algorithm is non-trivial and combines insights into the structure of a hypothetical solution with branching techniques, linear programming subroutines and insights from multidimensional geometry.

\smallskip
\noindent \textbf{Related Work.}\quad
Bei and Suksompong~(\citeyear{bei2021dividing}) studied graph cutting under the fairness notions of proportionality and equitability; this was the first paper that considered a graph structure with {\em divisible} items.
For indivisible items there are several different graph-based approaches. In the most common modelling scenario the items correspond to the vertices of the graph and each agent must get \iflong 
a connected subgraph~\cite{bei2019connected,bilo2021almost,BouveretCEIP17,DeligkasEGHO21,ESS21graphical,igarashi2019pareto,suksompong2019fairly}. 
\fi
\ifshort
a connected subgraph~\cite{bei2019connected,bilo2021almost,BouveretCEIP17,DeligkasEGHO21,ESS21graphical,igarashi2019pareto}. 
\fi

A different line of work uses graphs to denote the relationships between the agents, where an agent compares their bundle only against the bundles of the agents they are connected with~\cite{AbebeKP17,aziz2018knowledge,BeiQZ17,DBLP:conf/soda/BeiSWZZZ20,chevaleyre2017distributed,EibenGHO20}.

\iflong
There is a long line of papers for classical cake cutting and its extensions.
\cite{edward1999rental} gave an alternative proof of existence of contiguous envy free solutions using Sperner's lemma; this proof was used in~\cite{deng2012algorithmic} to place the approximate version of the computational problem in PPAD. Recently,~\cite{DBLP:journals/tamm/Segal-HaleviS21} extended the existential result for groups of agents and~\cite{DBLP:conf/ijcai/HosseiniIS20,IM21-layered} extended the result for multilayered cakes. Polynomial-time algorithms that compute approximately envy-free contiguous allocations were presented in~\cite{arunachaleswaran2019fair,deng2012algorithmic,goldberg2020contiguous}, while~\cite{seddighin2019expand} provides a polynomial-time algorithm for the problem for a very restricted class of valuations; to the best of our knowledge this is the only positive result for exact solutions. On the other hand, non-contiguous envy free cake cutting is easy to solve~\cite{cohler2011optimal,kurokawa2013cut,lipton2004approximately}.
\fi

\ifshort
In addition to the above, there are many other works that study variants of cake cutting~\cite{ElkindSS21separation,ElkindSS21land,segal2016waste,MenonL17,marenco2014envy,caragiannis2011towards,bei2021price,balkanski2014simultaneous,aumann2015efficiency,branzei2015dictatorship}.
\fi
\iflong
In addition to the above, there are many other works that study variants of cake cutting~\cite{ElkindSS21separation,ElkindSS21land,segal2016waste,MenonL17,marenco2014envy,caragiannis2011towards,branzei2017query,branzei2017communication,branzei2015dictatorship,bei2021price,balkanski2014simultaneous,aumann2015efficiency}.
\fi

\section{Preliminaries and Problem Definition}

\noindent \textbf{Notation.}\quad
For rational numbers \(i,j \in \mathbb{Q}\), we use the standard notation \([i,j] = \{k \in \mathbb{Q} \mid i \leq k \leq j\}\) for intervals and for an interval \(I = [i,j]\) we denote its length by \(|I| = \max\{j-i, 0\}\).
We denote the set of all non-negative rational numbers by \(\mathbb{Q}^+\).
\iflong
Further, we refer to the handbook by~\cite{Diestel12} for standard graph terminology.
\fi
Throughout the paper we will consider simple undirected graphs.

\smallskip
\noindent \textbf{Graph Cutting.}\quad
Consider a set $A$ of $n$ agents, a connected graph \(G = (V,E)\), and for each agent \(a \in A\) an edge weight function \(\val_a : E \to \mathbb{Q}^+\) over the edges of \(G\); \(\val_a\) is the \emph{utility} (function) of agent \(a\), and for a specific edge \(e \in E\) we call \(\val_a(e)\) the \emph{utility} of \(a\) for \(e\).

A \emph{piece of an edge} \(e\) is a tuple \((e, I)\) where \(I \subseteq [0,1]\) is a possibly empty interval.
We assume an arbitrary but fixed order \(V\) of the vertices of \(G\), and say that pieces \((e,I)\) and \((e',I')\) of two different edges \(e,e' \in E\) are \emph{adjacent} if
\begin{itemize}
	\item there is some \(v \in V\) such that \(e = uv\) and \(e' = vu'\) (i.e.\ \(e\) and \(e'\) are adjacent in \(G\)), and
	\item if \(u\) has a smaller index in the ordering of \(V\) than \(v\), then \(1 \in I\) and \(0 \in I\) otherwise, and
	\item similarly, if \(u'\) has a smaller index in the ordering of \(V\) than \(v\), then \(1 \in I'\) and \(0 \in I'\) otherwise.
\end{itemize}
A \emph{piece} of \(G\) is a collection \(P\) of pieces of edges of \(e\) such that for every pair of pieces of edges \((e_0,I_0), (e_\ell,I_\ell)\) in \(P\) there is a sequence of pieces of edges \((e_1,I_1), \dotsc, (e_{\ell-1},I_{\ell-1})\) in \(P\) such that for all \(j\) with \(0 \leq j < \ell\), \((e_j,I_j)\) and \((e_{j + 1},I_{j + 1})\) are adjacent.
An example is provided in Figure~\ref{fig:pieces}.

The utility of an agent \(a\) for a piece \(P\) is given as \(\val_a(P) = \sum_{(e,I) \in P}|I| \cdot \val_a(e)\).
Note that \(\{(e,[0,1]) \mid e \in E\}\) is also a piece of \(G\).
As is standard, our algorithms will assume \emph{normalized} utilities, i.e., \(\val_a(G) = 1\) for all \(a \in A\)\footnote{Instances with non-normalized utilities can be trivially transformed into equivalent ones with normalized utilities.}.

\begin{figure}
		\includegraphics[scale=0.95]{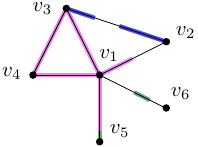}
		\begin{minipage}[b]{0.58\columnwidth}\small
			Edge pieces:\\
			pink:\((v_1v_2,[0,\frac{1}{2}])\),\((v_1v_3,[0,1])\), \((v_1v_4,[0,1])\),~\((v_1v_5,[0,\frac{5}{6}])\),~\((v_3v_4,[0,1])\);\\
			blue: \((v_2v_3,[0,\frac{1}{2}])\), \((v_2v_3,[\frac{3}{4},1])\);\\
			green: \((v_1v_5,[\frac{5}{6},1])\), \((v_1v_6,[\frac{1}{2},\frac{3}{4}])\)\\
			Only the pink edge pieces form a piece of the graph together; they are adjacent via \(v_1\), $v_3$ and $v_4$.
		\end{minipage}
		\caption{Examples of pieces of edges and of a graph.\label{fig:pieces}}
\end{figure}

A \emph{partition} of \(G\) into pieces is a set \(\Pi\) of pieces such that for every edge \(e \in E\) and every real $0\leq \alpha\leq 1$, there is precisely one piece \(P \in \Pi\) and one $(e,I)\in P$ such that $\alpha\in I$.
In some cases it is also necessary to allocate each vertex to a single piece: a partition of $G$ into pieces is \emph{vertex-disjoint} if all pieces of edges containing a vertex belong to the same piece of the graph.
See Figure~\ref{fig:partitions} for an example.
\begin{figure}
	\vfil
	\includegraphics[page=2,scale=0.95]{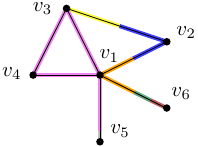}
	\begin{minipage}[b]{0.6\columnwidth}
		\small
		This partition is not vertex-disjoint as \(v_1\) is contained in both the orange and pink piece of the graph.
		Already if one orange edge piece was pink, the remaining orange edge piece could be given by an open interval, making the partition vertex-disjoint.
	\end{minipage}
\caption{Non-vertex-disjoint partition of a graph into seven connected pieces of different colors.\label{fig:partitions}}
\end{figure}

Finally we are ready to define our problem of interest. 
In (piecewise) \problem\ (\problemshort) we are given agents \(A\), graph \(G\) and utilities \(\val_a : E \to \mathbb{Q}^+\) for each agent \(a \in A\).
Our task is to construct a partition $\Pi$ of $G$ into pieces and a bijection (called an \emph{assignment}) \(\pi : A \to \Pi\) such that for every pair of agents \(a,a' \in A\), it holds that \(\val_a(\pi(a)) \geq \val_a(\pi(a'))\).
This condition is commonly referred to as \emph{envy-freeness} in assignment and allocation problems, and when it is violated for some \(a\) and \(a'\), we say that \(a\) \emph{envies} \(a'\).

We can analogously define the problem of (piecewise) \problemvd\ (\problemvdshort) by additionally requiring $\Pi$ to be vertex-disjoint.
In fact, by replacing each vertex in an instance of \problemshort\ with a clique of size $|E(G)|$, we obtain a simple reduction from the former problem to \problemvdshort.

\begin{observation}
\label{obs:reduce}
	\problemshort\ can be reduced to \problemvdshort\ in polynomial time. 
	\end{observation}

\smallskip
\noindent \textbf{Bounding the Number of Cells in Metric Spaces.}\quad
One prominent tool our main algorithm for solving \problemshort\ and
\problemvdshort\ uses is a theorem that applies to the behavior of
polynomials in higher-dimensional spaces. To provide a high-level
intuition,
consider a $d$-dimensional space that is cut into regions by $s$-many
hyperplanes or, more generally, ``well-behaved cuts'' defined by
bounded-degree polynomials.
The combination of these cuts splits the whole space into ``cells'', each consisting of points that lie on the same side of each of the $s$-many cuts. While the trivial bound for the number of these cells is $2^s$, it can be shown that the number of such cells is in fact polynomial in $s$ for fixed $d$ and that representatives of these cells can be computed efficiently. 
\ifshort
The result we use here is formalized in the book of \citeauthor{BasuPR06}~(\citeyear{BasuPR06}, \iflong Theorem\fi \ifshort Thm.\fi 13.22), see also Simonov et al.~(\citeyear{SimonovFGP19}). ($\star$)
\fi

\iflong
Let $V=\mathbb{R}^d$ and let us denote the ring of polynomials over variables $X_1$, \dots,
$X_d$ with coefficients in $\mathbb{R}$ by $\mathbb{R}[X_1, \dots,
X_d]$.
For a set of $s$ polynomials $\mathcal{P} = \{P_1, \dots, P_s\} \subset \mathbb{R}[X_1, \dots, X_d]$, a sign condition is specified by a sign vector $\sigma \in \{-1, 0, +1\}^s$, and the sign condition is non-empty over $V$ with respect to $\mathcal{P}$ if there is a point $x \in V$ such that
$\sigma = (\sign(P_1(x)), \dots, \sign(P_s(x))),$
where $\sign(\alpha)$ is the sign function defined as
$$\sign(\alpha) = \begin{cases}
	1, \text{ if } \alpha > 0,\\
	0, \text{ if } \alpha = 0,\\
	-1, \text{ if } \alpha < 0.
\end{cases}$$

\noindent The \emph{realization space} of $\sigma \in \{-1, 0, +1\}^s$ over $V$ is the set
$$R(\sigma) = \{x | x \in V, \sigma = (\sign(P_1(x)), \dots, \sign(P_s(x)))\}.$$
If $R(\sigma)$ is not empty then each of its non-empty semi-algebrically connected (which is equivalent to just connected on semi-algebraic sets~\cite[Theorem 5.22]{BasuPR06}) components is a \emph{cell} of $\mathcal{P}$ over $V$.

The following result, sometimes called Warren's Theorem~\cite{Waren68}
, gives an algorithm to compute a point in each cell of $\mathcal{P}$ over $V$. 

\begin{proposition}[Basu, Pollack and Roy~(\citeyear{BasuPR06}, Theorem 13.22), see also Simonov et al.~(\citeyear{SimonovFGP19})]
  Let
  $\mathcal{P} \subset \mathbb{R}[X_1, \dots, X_d]$ be a finite set of
  $s$ linear polynomials, i.e., each $P \in \mathcal{P}$ is of degree at
  most $1$.
  Let $D$ be a
  ring generated by the polynomials in
  $\mathcal{P}$.
  There is an algorithm using at most $\bigoh(s^d)$ algorithmic operations in $D$,
  which given $\mathcal{P}$ computes a set of points meeting every
  non-empty cell of $V$ over $\mathcal{P}$. 
  \label{prop:sample_points}
\end{proposition}
We remark Proposition~\ref{prop:sample_points} is merely the special
case of Basu, Pollack and Roy~(\citeyear{BasuPR06}, Theorem 13.22)
restricted to $V=\mathbb{R}^d$ and linear polynomials.
\fi

\section{Instances with Few Agents}

In this section we consider the complexity of \problemshort{} and
\problemvdshort{} for instances with only a few agents. Interestingly,
we will show that while both problems are \NP-hard even for instances
with only two agents, \problemshort{} turns out to be significantly
harder when additional restrictions are considered for the input
graph. In particular, while \problemshort{} is already \NP-hard on
stars, \problemvdshort{} can be solved in polynomial-time (for a fixed
number of agents) even on trees and only becomes \NP-hard on graphs
that have a vertex deletion set into a matching. Moreover, the problem becomes much harder if we relax the graph structure from trees to ``tree-like graphs'': both problems become \NP-hard on graphs that have treewidth $2$ and maximum degree $3$ (see Figure~\ref{fig:twNPh} for an illustration). 

\ifshort
All three \NP-hardness results follow from polynomial-time reductions from the \NP-complete \NPP{} problem: given a
  multi-set $S=\{s_1,\dots,s_n\}$ of non-negative integers, decide whether there is a partition of $S$ into two subsets $S_1$
  and $S_2$ such that $\sum_{s\in S_1}s=\sum_{s\in S_2}s$.
\fi

\iflong\begin{theorem}\fi
\ifshort\begin{theorem}[$\star$]\fi
\label{the:NP-ps}
  \problemshort{} is \NP-hard even when restricted to instances with
  two agents where the graph is a star.
\end{theorem}
\iflong
\begin{proof}
  We provide a simple polynomial-time reduction from the well-known
  \NP-complete \NPP{} problem. Let $S$
  be an instance of \NPP{}, i.e., $S$ is a
  multi-set of non-negative integers and the aim is to
  decide whether there is a partition of $S$ into two subsets $S_1$
  and $S_2$ such that $\sum_{s\in S_1}s=\sum_{s\in S_2}s$. Given $S$,
  we construct an instance $(A,G,(u_a)_{a\in A})$ of \problemshort{}
  by setting $A=\{a_1,a_2\}$, $G$ is the star with center $c$ and with
  one leaf $l_s$ for
  every $s \in S$, and $u_{a_1}((c,l_s))=u_{a_2}((c,l_s))=s$ for every
  $s \in S$. Observe that w.l.o.g. we can assume that every edge is
  fully assigned to one of the agents, since otherwise one of the
  agents would be assigned less than one edge, which would not result
  in an envy-free assignment. Therefore, every envy-free assignment
  results in a $2$-partition of the edges both having the same total
  valuation, which shows that $S$ is a
  yes-instance of \NPP{} if and only if $(A,G,(u_a)_{a\in A})$ is a
  yes-instance of \problemshort{}.
\end{proof}
\fi

\ifshort
\begin{theorem}[$\star$]
\fi
\iflong
\begin{theorem}
\fi
\label{the:matchNPhard}
  \problemvdshort{} is \NP-hard even when restricted to instances with
  two agents where the graph consists of a matching plus two additional
  vertices.
\end{theorem}
\iflong
\begin{proof}
  The reduction is similar to the one employed in
  Theorem~\ref{the:NP-ps} and uses a polynomial-time
  reduction from the \NPP{} problem.
  Let $S$
  be an instance of \NPP{}\iflong, i.e., $S$ is a
  multi-set of non-negative integers and the aim is to
  decide whether there is a partition of $S$ into two subsets $S_1$
  and $S_2$ such that $\sum_{s\in S_1}s=\sum_{s\in S_2}s$\fi. Given $S$,
  we construct an instance $(A,G,(u_a)_{a\in A})$ of \problemvdshort{}
  as follows. We set $A=\{a_1,a_2\}$ and the graph $G$ has two
  vertices $c_1$ and $c_2$ as well as two vertices $l_s$ and $l_s'$
  for every $s \in S$. Moreover, $G$ has all edges between
  $\{c_1,c_2\}$ and $\{l_s\mid s \in S\}$ and an edge between $l_s$
  and $l_s'$ for every $s \in S$. Note that $G-\{c_1,c_2\}$ is a
  matching. Finally, for every $i \in \{1,2\}$, we set $u_{a_i}(e)=s$ if $e=(l_s,l_s')$ for some $s \in
  S$ and $u_{a_i}(e)=0$, otherwise. Note that w.l.o.g. we can assume
  that every edge $(l_s,l_s')$ is fully assigned to exactly one of the
  agents since otherwise one agent would only be assigned a part of
  $(l_s,l_s')$ (and nothing else). Moreover, by assigning $c_i$ to
  agent $a_i$, we can obtain an assignment that assigns any subset of
  the edges $(l_s,l_s')$ to agent $a_i$ and all remaining such edges
  to agent $a_{2-i+1}$. Since every such assignment is envy-free if
  and only if the total valuations of the edges assigned to both
  agents is the same, we obtain that $S$ is a
  yes-instance of \NPP{} if and only if $(A,G,(u_a)_{a\in A})$ is a
  yes-instance of \problemvdshort{}.
\end{proof}
\fi

\ifshort
\begin{theorem}[$\star$]
\fi
\iflong
\begin{theorem}
\fi
\label{thm:twNPh}
  \problemvdshort{} and \problemshort{} are \NP-hard even when restricted to instances with
  two agents \iflong with identical valuations \fi where the graph has treewidth 2 and maximum degree 3.
\end{theorem}
\ifshort
\begin{proof}[Proof Sketch]
Given an instance of \NPP{} with $n$ positive integers, we create a graph with $4n$ vertices as follows (see also Fig.~\ref{fig:twNPh}). For every $i \in [n]$, we construct the vertices $a_i, b_i, c_i$, and $d_i$ and the edges $(a_i,c_i), (c_i,b_i)$ and $(c_i,d_i)$. In addition, for every $i \in [n-1]$ we create the edges $(a_i, a_{i+1})$ and $(b_i,b_{i+1})$. We create two agents with identical valuations: for every $i \in [n]$, both agents value the edge $(c_i,d_i)$ with $s_i$, while every other edge has value 0.
\end{proof}
\fi

\iflong
\begin{proof}
Again, we will reduce from \NPP{}. This time, the graph we will construct has maximum degree 3 and it is not hard to verify that has treewidth 2. We will use the same construction to prove both results, this is depicted at Figure~\ref{fig:twNPh}. More formally, given an instance of \NPP{} with $n$ positive integers, we create a graph with $4n$ vertices as follows. For every $i \in [n]$, we construct the vertices $a_i, b_i, c_i$, and $d_i$ and the edges $(a_i,c_i), (c_i,b_i)$ and $(c_i,d_i)$. In addition, for every $i \in [n-1]$ we create the edges $(a_i, a_{i+1})$ and $(b_i,b_{i+1})$. To complete the construction, we need to specify the valuations of the agents. We create two agents with identical valuations. For every $i \in [n]$, both agents value the edge $(c_i,d_i)$ with $s_i$. Every other edge has value 0.

To prove the theorem we observe the following. Firstly, note that w.l.o.g. we can assume that every edge $(c_i,d_i)$ is fully assigned to exactly one of the agents since otherwise one agent would only be assigned a part of $(c_i,d_i)$ (and nothing else). In addition, note that for \problemshort{} there is no allocation where a vertex is shared by the two agents. To complete our proof, we observe that every allocation corresponds to a partition of the $n$ integers, we can conclude that there exists an envy-free allocation if and only if there is a solution for the \NPP{} instance.
\end{proof}
\fi
\begin{figure}
\vfil
\includegraphics[page=2,scale=1.15]{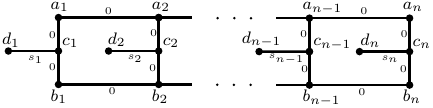}
\caption{The graph used in the proof of Theorem~\ref{thm:twNPh}.\label{fig:twNPh}}
\end{figure}

Next, we move on to the aforementioned algorithmic result for \problemvdshort. To establish that, we first prove the following technical lemma.

\iflong
\begin{lemma}
\fi
\ifshort
\begin{lemma}[$\star$]
\fi
\label{lem:reduction_to_ILP}
	 Let \(\mathcal{I}=(A,G,\{u_a\}_{a\in A})\) be an instance of \problemshort\ (\problemvdshort) and let $G$ be a tree, or a cycle, and let $F$ be a set of edges of $G$. There is an algorithm that runs in time polynomial in $|A|^{|F|+1}\cdot |F|^{|A|} \cdot |\mathcal{I}|$, \iflong where \(|\mathcal{I}|\) is the encoding length of \(\mathcal{I}\),\fi and either outputs an assignment \(\pi\) that is envy-free such that each connected component of $G-F$ is assigned to exactly one agent (and, for \problemvdshort, $\pi(A)$ is additionally a vertex-disjoint partition)  
	 or correctly identifies that such an assignment does not exist. 
\end{lemma}

\ifshort
\begin{proof}[Proof Sketch]
Observe that $G - F$ consists of most $|F|+1$ connected components and that there are at most  $|A|^{|F|+1}$ many assignments of these connected components to agents. By careful branching, we can determine for each agent whether they are (1) assigned to a specific single edge of $F$ and their piece is fully inside some $e \in F$, or (2) are assigned all edges in a subgraph $T$ of $G$ induced by a union of some connected component of $G-F$ and edges of $F$. For each branch, we design an instance of Linear Programming (LP) and solve it; if the instance has a solution, we can translate it into an assignment $\pi$ with the desired properties, and otherwise no such assignment exists for the current branch.
\end{proof}
\fi

\iflong
\begin{proof}
Observe that $G - F$ consists of most $|F|+1$ connected components.
In addition, observe that there are at most  $|A|^{|F|+1}$ many assignments of these connected components to agents. Observe that there might exist agents that have been assigned more than one component and other agents that might {\em not} have been assigned any component. Consider now one of these assignments. For each agent $a$ that is assigned more than one component, if we can, we make their component connected by assigning a
minimal set of edges from $F$ that makes $\pi(a)$ connected. Note, that if \(G\) is a tree this choice is unique, if it exists.
If \(G\) is a cycle we branch on up to two such choices if they exist.
If we cannot make the piece of every agent connected, then we discard this assignment and we proceed to the next one.
Let $F' \subseteq F$ be the set of edges not yet assigned to any agent, after we convert the components every agent gets into a single one. Each agent that is not assigned any piece yet, will get a piece {\em fully inside} some edge of $F'$. For each such agent we branch on the edge they will be contained in. This results in at most $|F'|^{|A|} \leq |F|^{|A|}$ branches. 

At this point there are two possibilities for an agent.
\begin{itemize}
\item They are assigned to a single edge of $F$ and its piece is fully inside some edge $e \in F$.
\item They get all edges in a subgraph $T$ of $G$ induced by a union of some connected component of $G-F$ and edges of $F$. In this case the agent is also allowed to get some pieces of edges incident with some vertex in $T$ (possibly not the full edge), however we do not branch on such edges for this agent.
\end{itemize}
 Note that if \(\mathcal{I}\) is an instance of \problemvdshort, then we can check that the trees $T_1$ and $T_2$ assigned to two different agents in our branching are vertex disjoint and terminate the branch if it is not the case. 

For each branch we now design an instance of Linear Programming (LP) that has a solution if and only if there exists an envy-free assignment consistent with our branching (which can be viewed as a ``guess'' of the properties of an optimal solution). 
For each edge $e\in F'$ and each agent $a\in A$ such that $a$ is allowed to have a piece inside the edge $e$ according to our guess we have a variable $x_e^a$ which represents the length of the piece of the edge $e$ that the agent $a$ gets. We add the following constraints into the LP:
\begin{align}
	&x_e^a\ge 0&  &\text{for all }a\in A, e\in F'\label{eq:No_agents_positive}\\
	&\sum_{a\text{ is allowed a piece of e}}x_e^a = 1& &e\in F'. \label{eq:No_agents_sum_edge}
\end{align}
Let $\mathbf{x}$ denote the vector of all variables in some fixed order. Given $\mathbf{x}$, we can compute the value of the piece assigned to an agent $a$. If agent $a$ is assigned a subtree $T$ and $F_T$ denotes the subset of edges in $T$ that are incident with some vertex in $T$, then for every agent $a'\in A$ we set  \[\val_{a'}(a, \mathbf{x})= \val_{a'}(T)+\sum_{e\in F_T}\val_{a'}(e)\cdot x_e^a.\] For an agent $a$ that is fully assigned a piece inside an edge $e\in F$ and every agent $a'$, we set \[\val_{a'}(a, \mathbf{x})= \val_{a'}(e)\cdot x_e^a.\]
Now for every ordered pair $(a,a')$ of agents $a,a'\in A$, we add the following constraint to the LP:
\begin{equation}
	\val_{a}(a, \mathbf{x})\ge \val_{a}(a', \mathbf{x}).\label{eq:No_agents_envy}
\end{equation} 
It is rather straightforward to see that if Constraints~(\ref{eq:No_agents_positive}) and (\ref{eq:No_agents_sum_edge}) are satisfied, then we can assign to each agent $a$ a connected piece of the graph such that the length of the piece the agent $a$ gets inside the edge $e$ is $x^a_e$ and if we guess the subgraph $T$ for the agent $a$, then $a$ is assigned all edges inside $T$ fully. Moreover, it is easy to see that in this case the utility of the agent $a$ is $\val_a(a,\mathbf{x})$. Hence, Constraint~(\ref{eq:No_agents_envy}) implies that such assignment is envy-free. Moreover, assuming an envy-free assignment \(\tempAsgn\) exists, one of our branches has to be consistent with \(\tempAsgn\) and in that case the algorithm will output a vertex-disjoint envy-free assignment. We conclude that if the algorithm does not find a solution in any of the branches then the input is a \texttt{No}-instance.

Overall this procedure requires time in \(\mathcal{O}(|A|^{|F| + 1} \cdot |F|^{|A|} \cdot T(|F||A| + |A|))\), where \(T(x)\) denotes the time required to solve a linear program with \(x\) variables.
The most efficient currently known procedure requires \(\mathcal{O}(x^\omega L)\)~\cite{Alman2021} where \(\omega\) is the matrix multiplication constant and \(L\) is the number of input bits, which in our case is upper-bounded by \(|\mathcal{I}|\).
\end{proof}
\fi

Lemma~\ref{lem:reduction_to_ILP} allows us to solve \problemvdshort\ on trees by applying initial branching to reach a situation satisfying the preconditions of Lemma~\ref{lem:reduction_to_ILP}.

\ifshort
\begin{theorem}[$\star$]\fi\iflong{}\begin{theorem}\fi~\label{thm:No_agents}
For each $c\in \mathbb{N}$, \problemvdshort\ restricted to instances with at most $c$ agents is polynomial-time solvable on trees.
\end{theorem}
\ifshort
\begin{proof}[Proof Sketch]
\fi
\iflong
\begin{proof}
\fi
Let \(\tempAsgn\) be an arbitrary vertex-disjoint envy-free
assignment. Since $G$ is a tree and agents do not share vertices,
\ifshort{}one can show\fi\iflong{}we claim\fi \ that there is a set
$F$ of at most $|A|-1$ edges with the following properties: 
(I) each of the exactly $|F|+1$
components of $G-F$ is assigned to a single agent, and (II) every edge $e\in
F$ is split between the two agents assigned to the components that
contain a vertex incident with~$e$ plus some of the agents who are not assigned any of
these components.

 \iflong{}To see the claim, let us
root $G$ in some arbitrary vertex $r$. Let $a_r$ be the unique agent
that contains the vertex $r$ in some piece of \(\tempAsgn(r)\). For every other agent $a$, there is a unique point $x$ in the closure of pieces
of $a$ such that the unique path from $x$ to $r$ does not intersect any piece of $a$ (note that $x$ is not necessarily a vertex and could be a point inside some edge). Given this point we assign to the agent $a$ a single edge $e_a$ as follows: 
	\begin{itemize}
		\item if $x$ is inside some edge $f$, then we let $e_a=f$,
		\item if $x$ is a vertex of $G$ and there is a unique edge $f$ that contains a piece of $a$, then we let $e_a=f$, and 
		\item else $x$ is a vertex of $G$ incident to at least $2$ edges that get a piece of $a$. We let $e_a$ be the unique edge incident on $x$ on the unique $x$-$r$ path in $G$. 
	\end{itemize}
	Let $F=\{e_a\mid a\in A\setminus \{a_r\}\}$. We show that each connected component of $G-F$ is assigned to a single agent. The rest of the claim then follows straightforwardly  from the connectivity of the pieces assigned to single agents. 
	Let $T$ be a subtree of $G$ induced by some connected
        component of $G-F$. We can root $T$ in the unique vertex that
        is ``closest'' to $r$. That is the vertex $r_T$ such that the
        unique $r_T$-$r$ path does not contain any vertex of $T$. We
        show that all edges of $T$ are assigned to some agent $a$ such
        that $e_a$ is the unique edge from $r_T$ to its parent in
        $G$. We distinguish two possibilities. If all agents contain a
        piece in at most one edge incident with $r_T$, then it is easy
        to see that all edges incident with $r_T$ are in
        $F$. Therefore, $T$ is a singleton and the claim
        follows. Else, if there is an agent $a$ such that $a$ gets
        pieces in at least two edges incident on $r_T$. Then $e_a$ is
        the edge from $r_T$ to its parent and $a$ is the unique such
        agent; i.e., all other agents get pieces in at most one edge
        incident on $r_T$. We show by induction on the distance from $r_T$ that all edges in $T$ are fully assigned to $a$ by $\tempAsgn$. For the base case observe that if some agent contains a piece of an edge $f$ incident with $r_T$, then $f\in F$ and in particular $f\notin E(T)$. Therefore all the edges incident of $r_T$ in $T$ are assigned to $a$ by $\tempAsgn$.
	Now let us consider a vertex $v\in V(T)$ at distance $i$ from
        $r_T$. We know, by the induction hypothesis, that the unique
        edge $f$ between $v$ and its parent in $T$ is fully assigned
        to $a$. If there is an agent $a'$ that contains pieces in at
        least $2$ edges incident o $v$, then $e_{a'}=f$, which is not
        possible as $f\in E(T)$. Therefore every agent $a'$ (other
        than $a$) contains a piece in at most one edge incident with $v$ and if $a'$ contains a piece inside edge $f'$ incident on $v$, then it is easy to see that $e_{a'}=f'$ and $f'\notin E(T)$. Hence all edges incident with $v$ in $T$ are also fully assigned to $a$ and the claim follows by induction. We note that the same discussion works even if $r_T=r$, in which case we show that $T$ is assigned to $a_r$ analogously.  	
	\fi    

We can now enumerate all of the at most $|V(G)|^{|A|}$ possible sets
of edges $F\subseteq E(G)$ such that $|F|\le |A|-1$. The theorem then follows by applying Lemma~\ref{lem:reduction_to_ILP} to $G$ and $F$.

Using the notation and the running time bound from the proof of Lemma~\ref{lem:reduction_to_ILP}, it follows that the overall complexity lies in \(\mathcal{O}(|V(G)|^{|A|} \cdot |A|^{|A|} \cdot (|A| - 1)^{|A|} \cdot T(|A|^2))\).
\end{proof}

To complete our understanding of \problemshort{} and \problemvdshort{} for instances with only boundedly-many agents, we show that both problems are also polynomial-time solvable on graphs of maximum degree $2$ (i.e., cycles) and that the former is polynomial-time solvable on bounded-degree trees. These results follow a similar approach as the proof of Theorem~\ref{thm:No_agents}, which is a combination of initial branching and Lemma~\ref{lem:reduction_to_ILP}.

\iflong \begin{theorem} \fi
\ifshort \begin{theorem}[$\star$]\fi
\label{thm:treedegreeP}
	For each $c,d\in \mathbb{N}$, \problemshort{} restricted to instances with at most $c$ agents is polynomial-time tractable on trees with maximum degree at most $d$.
\end{theorem}

\iflong
\begin{proof}
	We show that there are at most $|A|$ many vertices and edges of $G$ than can be in a piece of more than one agent. Given this claim, we can enumerate all of the at most $(2\cdot|V(G)|-1)^{|A|}$ possible sets
	of vertices $U\subseteq V(G)$ and edges $F\subseteq E(G)$ such that $|U\cup F|\le |A|$. Let $F'$ be the union of $F$ and the edges incident with vertices in $U$. Note that $|F'|\le |A|\cdot \Delta$, where $\Delta$ is the maximum degree of a vertex in $G$. The theorem then follows by applying Lemma~\ref{lem:reduction_to_ILP} to $G$ and $F'$ in each branch. 
	
	The rest of the proof is devoted to showing that there are at most $|A|$ many vertices and edges of $G$ than can be in a piece of more than one agent.
	Let \(\tempAsgn\) be an arbitrary envy-free assignment. 	
	We start the proof by showing that there are at most $|A|$ many vertices and edges of $G$ than can be in a piece of more than one agent. Let $U$ be the set of such vertices and $F$ the set of such edges. That is for every vertex $u\in U$ there exists 
		two different agents $a_1$ and $a_2$ such that both \(\tempAsgn(a_1)\) and \(\tempAsgn(a_2)\) contain a piece that contain $u$. Similarly, for every edge $e\in F$ there are two agents $a_1$ and $a_2$ such that both \(\tempAsgn(a_1)\) and \(\tempAsgn(a_2)\) contain a piece of the edge $e$.
		We will construct an injective function $f\colon U\cup F\rightarrow A$, that is we show that we can assign to every element (vertex or edge) in $U\cup F$ a unique agent. Let us root $G$ in some arbitrary vertex $r$. For the rest of the proof, for a vertex $v$, we will denote by $p_v$ the parent of the vertex $v$ in the rooted tree $G$ rooted in $r$. 
	For a vertex $u\in U$, at most one agent contains a piece of the edge $up_u$ from $u$ to its parent $p_u$ that contains the vertex $u$. Since there are at least two agents whose pieces contain the vertex $u$, there is at least one other agent $a$ such that $a$ does not get a piece of the edge $up_u$, but gets a piece of some other edge that contains the vertex $u$. We let $f(u)=a$.
	For an edge $e\in F$ let $f = up_u$, where $p_u$ is the parent of $u$. We let $f(e)= a$, where $a$ is the agent that gets the piece containing $u$. Note that since at least two pieces contain a piece of the edge $e$, it follows that the piece of agent $a$ does not contains $p_u$.
	
	Now let $x$ and $y$ be two different elements of $U\cup F$. We show that $f(x)\neq f(y)$.
	Clearly, if $f(x)=f(y) = a$, then the piece assigned to the agent $a$ under \(\tempAsgn\) contains a piece of both $x$ and $y$. We will show that if an agent $a$ such that $\tempAsgn(a)$ contains piece of both $x$ and $y$, then either $f(x)\neq a$ or $f(y)\neq a$.
	Let $u$ denote the vertex $x$ if $x\in U$ otherwise let $u$ be the vertex such that $x$ is the edge $up_u$ in $F$. Similarly, let $v$ be $y$ if $y\in U$ else let $v$ be the vertex such that $y$ is the edge $vp_v$. Note that if $f(x)=a$, then $\tempAsgn(a)$ contains a piece containing the vertex $u$. Similarly, if $f(y)=a$, then $\tempAsgn(a)$ contains $v$. Since piece of the agent $a$ is connected, it follows that $\tempAsgn(a)$ contains a path from $u$ to $v$, moreover since $G$ is a tree, this path is unique.
	Now note that if $v$ is an ancestor of $u$, then the unique path from $u$ to $v$ contains the edge $up_u$ and from the definition of $f$ it follows $f(x)\neq a$. If $u$ is an ancestor of $v$, then the unique path from $u$ to $v$ in $G$ contains the edge $vp_v$, and $f(y)\neq a$. Finally, if $v$ is not ancestor of $u$ and $u$ is not an ancestor of $v$, then the unique path from $u$ to $v$ contains both edges $up_u$ and $vp_v$, hence $f(x)\neq a$ and $f(y)\neq a$. 
	
	Using the notation and the running time bound from the proof of Lemma~\ref{lem:reduction_to_ILP}, it follows that the overall complexity lies in \(\mathcal{O}((2|V(G)| - 1)^{|A|} \cdot |A|^{|A|\Delta + 1} \cdot (|A|\Delta)^{|A|} \cdot T(|A|\Delta|A| + |A|))\).
\end{proof}
\fi

\ifshort \begin{theorem}[$\star$]\fi 
\iflong \begin{theorem} \fi
\label{thm:cycles}
	For each $c\in \mathbb{N}$, \problemshort\ and \problemvdshort\ restricted to instances with at most $c$ agents is polynomial-time solvable on cycles.
\end{theorem}

\iflong
\begin{proof}
	It is straightforward to see that in any partition of a cycle graph $G$ into $|A|$ connected pieces there are at most $|A|$ edges that are assigned to two different agents.
	The theorem then follows by branching on all $|V(G)|^{|A|}$ many possibilities of at most $|A|$ many edges shared by more than one agent. The theorem then follows by applying Lemma~\ref{lem:reduction_to_ILP} in each branch. 
	
	Using the notation and the running time bound from the proof of Lemma~\ref{lem:reduction_to_ILP}, it follows that the overall complexity lies in \(\mathcal{O}(|V(G)|^|A| \cdot |A|^{2|A| + 1} \cdot T(|A|^2 + |A|))\).
\end{proof}
\fi

\section{Instances with Few Edges}

In the section we provide an algorithm showing that both \problemshort\ and \problemvdshort\ are in \XP\ parametrized by the number of edges in the input graph $G$.

	\begin{theorem}\label{thm:No_edges_main}
		\problemshort\ and \problemvdshort\ can be solved in time \(|A|^{\bigoh(|E(G)|^2)}\).
	\end{theorem}

The algorithm can be divided into three main steps. We start with a direct brute-force branching over all assignments of agents that span more than one edge, and for these special agents we identify precisely the edges from which they will receive a piece. We also branch to determine the exact number of agents that will be assigned to each edge. This will result in $|A|^{\bigoh(|E(G)|)}$ many initial branches, and each branch already provides useful information about a hypothetical sought-after solution---but not enough to solve the problem. 
Crucially, every solution to the original instance corresponds to one of the branches. 

Our aim in the second step will be to construct, for each branch, a Linear Program (LP) to determine the exact lengths of all the pieces in an envy-free partitioning. In particular, if the branch corresponds to a solution, then we require that the LP outputs a partitioning that can be matched to agents in a way which also produces a solution. Unfortunately, the branching carried out in the previous step is not yet sufficient to construct such an LP: during the construction, we need to apply an additional advanced branching step to identify a small number of \emph{envy-critical} agents that are assigned completely to a single edge. The property of these agents is that they will be the ``closest'' to envying agents assigned to other edges in the graph, and in the LP these will serve as anchors which ensure that an envy-free assignment will exist as long as the assignment of agents is carried out in a way which respects the selected envy-critical agents. Defining, bounding the number of, and branching on these envy-critical agents is the most challenging part of the algorithm, and is also where 
\iflong Proposition~\ref{prop:sample_points} will be used.
\fi 
\ifshort Theorem 13.22~\cite{BasuPR06} is used.
\fi

Finally, based on the branching decisions and a solution to the constructed LP, we design an instance of bipartite matching that matches the remaining unassigned agents with the pieces given by the LP instance. If a matching exists we are guaranteed to have found a solution; if not, then our branch does not correspond to a valid solution.

\smallskip
\noindent \textbf{Initial Branching.}\quad
For the remainder of this section, we fix an instance \(\mathcal{I}\) of \problemshort\ or \problemvdshort\ given by the set of agents \(A\), graph \(G=(V,E)\), and utilities \(\val_a : E \to \mathbb{Q}^+\) for each agent \(a \in A\). Denote $k=|E(G)|$.
We can now start with the branching phase. Let us assume $\III$ admits an envy-free assignment \(\tempAsgn\) into some partition of $G$ into pieces; we will describe the branching as a series of ``guesses'' of the properties of this solution $\III$ and its interactions with $G$. 

First, observe that for each edge $e$ there are at most $2$ agents that can be assigned a piece of the edge $e$ together with a piece of some other edge. These are the agents in \(\tempAsgn\) that receive the piece $(e, [0,c])$ and the piece $(e, [d,1])$ for some constant $c,d\in [0,1]$. For each edge $e$, we guess the agent that gets the piece $(e,[0,c])$ (for some unspecified $c\in [0,1]$) and say that this is the guess for pair $(e,0)$; analogously, we guess the agent that gets the piece $(e,[d,1])$ (for some unspecified $d\in [0,1]$) and say that this is the guess for pair $(e,1)$. This results in $|A|^{2k}$ many branches. Let $A_V$ be the set of the at most $2k$ agents guessed in the previous step. 

All the remaining agents are assigned by \(\tempAsgn\) a piece $\{(e,[c,d])\}$ for some edge $e$. For every such agent, we say that it gets a piece \emph{fully contained inside edge $e$}. While it is too computationally expensive to guess precisely which agents get a piece fully contained in an edge $e$, we will guess the number of agents that get a piece fully contained in an edge $e$. This results in $|A|+1$ many guesses for each edge, amounting to a branching factor of at most $(|A|+1)^k$. Let us denote by $n_e$ the number of agents that get a piece fully contained in the edge $e$. We now perform a set of sanity checks on our branching; in particular, we discard branches which do not fulfil the following conditions:

\begin{enumerate}
	\item $|A_V|+\sum_{e\in E(G)} n_e=|A|$.
	\item For every agent $a\in A_V$, the guesses of pieces assigned to $a$ form a connected subset of $G$.
		 More formally, whenever our branch assigns an agent $a$ to two distinct pairs $(e,i)$ and $(f,j)$, where $e,f\in E(G)$ and $i,j\in \{0,1\}$, there must exist a path $P=e_1e_2\ldots e_q$ from $e[i]$ to $f[j]$ such that for each $\iota\in [q]$ it holds that ($\texttt{I}$) $n_{e_{\iota}}=0$ and ($\texttt{II}$) the agent $a$ is also the guess for both $(e_{\iota},0)$ and $(e_{\iota},1)$.
		\item In the case of \problemvdshort, we will also verify that the branching corresponds to a vertex-disjoint partition. In particular, for each vertex $v$ let $E_{v,0}$ be the set of edges incident to $v$ and a vertex preceding $v$ in the ordering and $E_{v,1}$ be the set of edges incident to $v$ and a vertex succeeding $v$ in the ordering. We check that there is a single agent $a$ such that for each edge $e\in E_{v,0}$, $a$ is the guess for $(e,0)$, and at the same time for each edge $e'\in E_{v,1}$, $a$ is the guess for $(e',1)$.
\end{enumerate}

\smallskip
\noindent \textbf{Linear Programming.}\quad
We can now begin describing the instance of LP that we will use to determine the partition. For this, it will be useful to observe that agents fully contained in the same edge must receive a segment of the same length.
\begin{observation}\label{obs:same_length_pieces_on_edge}
	Let $e\in E(G)$ be an edge and $a_1,a_2\in A$ two agents such that \(\tempAsgn(a_1)= \{(e,[x_1,y_1])\}\) and \(\tempAsgn(a_2)= \{(e,[x_2,y_2])\}\), then $|y_1-x_1| = |y_2-x_2|$. 
\end{observation}
\iflong 
\begin{proof}
	Since $u_{a_1}(\tempAsgn(a_1)) = |y_1-x_1| \cdot u_{a_1}(e)$ and $u_{a_1}(\tempAsgn(a_2)) = |y_2-x_2| \cdot u_{a_1}(e)$, 
		envy-freeness implies $|y_1-x_1| \ge |y_2-x_2|$. Similarly, $u_{a_2}(\tempAsgn(a_1)) = |y_1-x_1| \cdot u_{a_2}(e)$ and $u_{a_2}(\tempAsgn(a_2)) = |y_2-x_2| \cdot u_{a_2}(e)$, hence envy-freeness implies $|y_1-x_1| \le |y_2-x_2|$.
\end{proof}
\fi 
For each edge $e\in E(G)$, the LP instance will have variables $x^0_e$, $\delta_e$, and $x^1_e$. The variable $\delta_e$ represents the length of each piece $(e, [c,d])$ assigned to any agent that gets a piece fully inside $e$. The variable $x^i_e$ represents the length of the piece of the edge $e$ that was assigned to the agent $a\in A_V$ which is the guess for pair $(e,i)$. \ifshort
We start by adding constraints ensuring that all pieces have non-negative length and that the sum of lengths of pieces on each edge is exactly $1$.
\fi
\iflong
We start by adding the following two constraints saying that all pieces have non-negative length and the sum of lengths of pieces on each edge is exactly $1$.  
\begin{align}
		x^0_e,\delta_e, x^1_e \geq 0 \label{eq:piece_length_positive}\\
	x^0_e+n_e\cdot\delta_e+x^1_e=1.\label{eq:piece_length_sum_to_1}
\end{align}
\fi

\noindent
For each agent $a\in A_V$, let $\agentPieces_a$ denote the set of pieces that $a$ is the guess for: \(\agentPieces_a=\{(e,i)\mid e\in E(G), i\in \{0,1\}, \text{ and } a \text{ is the guess for } (e,i) \}\).
Given the intended meaning of variables $x^0_e$, $\delta_e$, and $x^1_e$, we can now add constraints to guarantee envy-freeness between agents in $a\in A_V$. 
For all $a,a'\in A_V$ we create the constraint
\begin{equation}
	\sum_{(e,i)\in \agentPieces_a}\val_a(e)\cdot x^i_e \ge \sum_{(e',i')\in \agentPieces_{a'}}\val_a(e')\cdot x^{i'}_{e'},\label{eq:vertex-vertex_envy}
\end{equation}
and for every $a\in A_V$ and $e\in E(G)$ we create the constraint  
\begin{equation}
	\sum_{(f,i)\in \agentPieces_a}\val_a(f)\cdot x^i_f \ge \val_a(e)\cdot \delta_e.\label{eq:vertex-edge_envy}
\end{equation}

\noindent
Next, for every (ordered) pair of edges $e,f$ such that $n_e>0$ and $n_f>0$, need to guarantee that agents that get a piece of length $\delta_e$ do not envy agents with pieces fully inside edge $f$. If we knew that $a\in A\setminus A_V$ gets a piece fully inside $e$, then for this agent we could express this via the constraint $\val_a(e)\cdot \delta_e\ge \val_a(f)\cdot \delta_f$.  
Unfortunately, we do not know which agents get a piece fully inside $e$ and cannot obtain this information by exhaustive branching in view of our time bounds.

To overcome this obstacle, let us order the agents in $A\setminus A_V$ by the ratio $\frac{\val_a(e)}{\val_a(f)}$ and consider two agents $a_1, a_2$ such that  
\iflong
$$\frac{\val_{a_1}(e)}{\val_{a_1}(f)}\ge \frac{\val_{a_2}(e)}{\val_{a_2}(f)}.$$
\fi
\ifshort
$\frac{\val_{a_1}(e)}{\val_{a_1}(f)}\ge \frac{\val_{a_2}(e)}{\val_{a_2}(f)}.$
\fi
It is easy to see that $\val_{a_2}(e) \cdot \delta_e\ge \val_{a_2}(f) \cdot \delta_f$ implies $\val_{a_1}(e) \cdot \delta_e \ge \val_{a_1}(f) \cdot \delta_f$. 
Hence to capture the desired constraint it will be sufficient to guess, for each ordered pair of edges $(e,f)$ the agent $a_{(e,f)}$ that is assigned a piece fully inside $e$ (by \(\tempAsgn\)), and has the smallest value for the fraction $\frac{\val_{a_{(e,f)}}(e)}{\val_{a_{(e,f)}}(f)}$ among all the agents that are assigned a piece fully inside $e$.
Intuitively, this corresponds to guessing an envy-critical agent $a$: among all the agents fully assigned to $e$, the agent $a$ is ``closest'' to envying agents fully assigned to $f$. Note that the guessed envy-critical agents will later preclude some agents from receiving a piece of $e$ (in particular, those that precede $a$ in the linear order defined by the fractions).

The procedure described above introduces at most $k^2$ many guesses of agents, which amounts to an additional branching factor of at most $|A\setminus A_V|^{k^2}$.
For each agent $a_{(e,f)}$, we then add the following constraint to the LP instance 
	\begin{equation}
		\val_{a_{(e,f)}}(e) \cdot \delta_e\ge \val_{a_{(e,f)}}(f) \cdot \delta_f.\label{eq:edge-edge_envy}
\end{equation}
We also perform additional consistency checks for this branching. First of all, we discard branches which select the same agent as being envy-critical in multiple pieces (i.e., if $e\neq e'$ then we require $a_{(e,f)} \neq a_{(e',f')}$\iflong , regardless of $f$ and $f'$\fi).
Moreover, since we have guessed at most $k-1$ agents for an edge $e$, we also check that the intended meaning of the choice of the agent $a_{(e,f)}$ is satisfied so far: for all triples of edges $e,f,f'$ we check that 
\begin{equation}
\tag{A}
\frac{\val_{a_{(e,f')}}(e)}{\val_{a_{(e,f')}}(f)}\ge \frac{\val_{a_{(e,f)}}(e)}{\val_{a_{(e,f)}}(f)}. \label{ineq:edge-edge}
\end{equation}

\noindent 
At this point we have added constraints which prevent---assuming our guesses were correct---an agent in $A_V$ from envying any other agent, and agents outside of $A_V$ from envying each other. Finally, for every edge $e$ and every agent $a\in A_V$ we would like to guarantee that the agents that get a piece fully inside $e$ do not envy the agent $a$. Similarly as before, for each specific agent $a'$ that gets a piece fully inside $e$ we could hypothetically ensure this via the constraint 
\iflong
$$ \val_{a'}(e)\cdot \delta_e \ge \sum_{(f,i)\in \agentPieces_a}\val_{a'}(f) \cdot x^i_{f}.$$ 
\fi
\ifshort
$ \val_{a'}(e)\cdot \delta_e \ge \sum_{(f,i)\in \agentPieces_a}\val_{a'}(f) \cdot x^i_{f}.$
\fi
However, we again do not know the agents that are assigned a piece fully inside $e$. Unfortunately, while for two edges $e$ and $f$ it was not too difficult to define and identify envy-critical agents and write linear constraints only for those, when comparing the envy of agents fully assigned to $e$ towards an agent $a\in A_V$ that receives multiple pieces of edges, the notion of ``envy-criticality'' we need depends on the size of the pieces $a$ gets from each edge. In particular, there is no fixed total ordering of the agents that allows us to define envy-criticality.
To give a concrete example of this issue, for two different
instantiations of the $x^i_f$ variables, say $x^i_f:=c^i_f$ and 
$x^i_f:=d^i_f$, and two agents $a_1$ and $a_2$ it may hold that
$$
\frac{\sum_{(f,i)\in \agentPieces_a}\val_{a_1}(f)\cdot c^i_f}{\val_{a_1}(e)}\ge \frac{\sum_{(f,i)\in \agentPieces_a}\val_{a_2}(f)\cdot c^i_f}{\val_{a_2}(e)}, 
$$
but 
$$
\frac{\sum_{(f,i)\in \agentPieces_a}\val_{a_1}(f)\cdot d^i_f}{\val_{a_1}(e)}\le \frac{\sum_{(f,i)\in \agentPieces_a}\val_{a_2}(f)\cdot d^i_f}{\val_{a_2}(e)}. 
$$

On the other hand, the assignment \(\tempAsgn\) does define some specific instantiation of $x^i_e$'s for which there is a (not necessarily strict) total ordering on the agents $a'$ in $A$ capturing how ``close'' they are to envying $a$, i.e., based on the value of
\iflong
$$\frac{\sum_{(f,i)\in \agentPieces_a}\val_{a'}(f)\cdot x^i_f}{\val_{a'}(e)}. $$
\fi
\ifshort
$\frac{\sum_{(f,i)\in \agentPieces_a}\val_{a'}(f)\cdot x^i_f}{\val_{a'}(e)}. $
\fi
While we have no way of computing which total ordering arises from the hypothetical assignment \(\tempAsgn\), we will later (in Lemma~\ref{lem:use-blackbox}) use
\iflong Proposition~\ref{prop:sample_points}
\fi 
\ifshort Theorem 13.22~\cite{BasuPR06} described in the Preliminaries
\fi
to show that only $|A|^{\bigoh(k)}$ many such orderings are possible, and moreover that we can enumerate all of these in time $|A|^{\bigoh(k)}$. For now, let us complete the description of the LP with this in mind. Since the number of relevant orderings is bounded by $|A|^{\bigoh(k)}$, we can apply branching to guess the ordering that arises from a hypothetical targeted assignment.
At that point we can also guess, for each edge $e$ and agent $a\in A_V$, the \emph{envy-critical} (according to this ordering) agent $\alpha_{e,a}$
that is assigned to the edge $e$ and envies the agent $a$ the most---and later use this guess to preclude some agents from being fully assigned to~$e$. 

For each guess, we add constraints to the LP which will ensure that the guess will be consistent with whatever solution the LP produces. In particular, we add the constraint
\begin{equation}
	\val_{a'}(e)\cdot \delta_e \ge \sum_{(f,i)\in \agentPieces_a}\val_{a'}(f) \cdot x^i_f \label{eq:edge-vertex_envy}
\end{equation}
for every agent $a'$ that envies $a$ at most as much as $\alpha_{e,a}$ according to the guessed ordering. More precisely, with each guess we will get some instantiation $x^i_e=y^i_e$ witnessing this guess, and we will insert a copy of Constraint~\ref{eq:edge-vertex_envy} for every agent $a'$ satisfying the following property:

\begin{equation}
\tag{B}
\frac{\sum_{(f,i)\in \agentPieces_a}\val_{\alpha_{e,a}}(f)\cdot y^i_f}{\val_{\alpha_{e,a}}(e)}\ge \frac{\sum_{(f,i)\in \agentPieces_a}\val_{a'}(f)\cdot y^i_f}{\val_{a'}(e)}. \label{ineq:edge-vertex}
\end{equation}

Our last task in this step is to provide a way to perform the branching over total orders described above. Recall that each instantiation of the variables $x_e^i$, for $i\in \{0,1\}$ and $e\in E(G)$, gives rise to a set of total orderings of all agents in $A$. 
In particular, 
each ordering is associated with precisely one pair $(a\in A_V, e\in E(G))$. Here, each ordering captures the relative envy towards $a$ under the assumption that the agents would be assigned to $e$ (see Inequality~\ref{ineq:edge-vertex}). We call a set of such orderings a \emph{portfolio}.
Moreover, since $\tempAsgn$ also corresponds to an instantiation of the variables $x_e^i$, it too gives rise to a set of total orderings, which we call a \emph{portfolio of} $\tempAsgn$. 

\iflong
\begin{lemma}
\fi
\ifshort
\begin{lemma}[$\star$]
\fi
\label{lem:use-blackbox}
It is possible to construct, in time $A^{\bigoh(k)}$, a set $\mathcal{R}$ of at most $k^kA^{\bigoh(k)}$-many portfolios which is guaranteed to contain the portfolio of $\tempAsgn$.
\end{lemma}

\iflong
\begin{proof}
Let $V=\mathbb{R}^{2k}$ be an algebraic set defined by the constant zero polynomial in variables $x^i_e$, for all $e\in E(G)$ and $i\in \{0,1\}$, and let us define the set \(\mathcal{P}\) of $k\cdot|A|^3$ many polynomials in these variables. Namely for each edge $e\in E(G)$, agent $a\in A_V$, and pair of agents $a_1,a_2\in A$, let $P^{(e,a)}_{(a_1,a_2)} =$ 
\[\frac{\sum_{(f,i)\in \agentPieces_a}\val_{a_1}(f)\cdot x^i_f}{\val_{a_1}(e)} - \frac{\sum_{(f,i)\in \agentPieces_a}\val_{a_2}(f)\cdot x^i_f}{\val_{a_2}(e)}. \]

Note that \[\frac{\sum_{(f,i)\in \agentPieces_a}\val_{a_1}(f)\cdot x^i_f}{\val_{a_1}(e)} > \frac{\sum_{(f,i)\in \agentPieces_a}\val_{a_2}(f)\cdot x^i_f}{\val_{a_2}(e)}\] if and only if \(\sign(P^{(e,a)}_{a_1,a_2}(\vec{x})) = 1\), \[\frac{\sum_{(f,i)\in \agentPieces_a}\val_{a_1}(f)\cdot x^i_f}{\val_{a_1}(e)} = \frac{\sum_{(f,i)\in \agentPieces_a}\val_{a_2}(f)\cdot x^i_f}{\val_{a_2}(e)}\] if and only if \(\sign(P^{(e,a)}_{a_1,a_2}(\vec{x})) = 0\), and \(\sign(P^{(e,a)}_{a_1,a_2}(\vec{x})) = -1\) otherwise. The assignment \(\tempAsgn\) gives concrete values for $x^i_e$'s and hence defines some point $w$ in $\mathbb{R}^{2k}$. Depending on the values of the polynomials  \(P^{(e,a)}_{a_1,a_2}\) in $w$, the point $w$ defines a sign vector \(\sigma\). The set of all points that define the same sign vector \(\sigma\) is a cell of \(\mathcal{P}\) over $\mathbb{R}^{2k}$; observe that each sign vector completely describes one possible ordering of the agents w.r.t. these inequalities (or, equivalently, Inequality~\ref{ineq:edge-vertex}).

Proposition~\ref{prop:sample_points} provides an algorithm that computes a point in each cell of \(\mathcal{P}\) over $V$ in time $k^k|A|^{\bigoh(k)}$, which immediately gives the same upper bound on the number of cells. Since $w$ is a point in some cell, one of the points identified by the algorithm is in the same cell as $w$ and defines the same sign vector \(\sigma\). 
\end{proof}
\fi

To formalize the description provided earlier, we now branch over all at most $k^k|A|^{\bigoh(k)}$ many portfolios obtained from Lemma~\ref{lem:use-blackbox}, or equivalently, points in the described metric space. Given some point $y\in \mathbb{R}^{2k}$, we guess for each pair of edge $e\in E(G)$ and agent $a\in A_V$ an agent $a_{e,a}$ as described above and introduce the LP instance constraints described in~(\ref{eq:edge-vertex_envy}). Finally, similarly as after introducing Constraints~(\ref{eq:edge-edge_envy}), we can again check that the intended meaning of guessed agents for each edge $e$ hold by checking Inequalities~(\ref{ineq:edge-edge})~and~(\ref{ineq:edge-vertex}) for every pair of agents 
guessed for each edge $e$. If at least one of the inequalities do not hold, then we reject the branch.
This finishes the construction of the LP instance. 

\smallskip
\noindent \textbf{Bipartite Matching.}\quad
Now, we can solve each LP instance in at most cubic time~\cite{CohenLS19}. If the instance is unsatisfiable, then the algorithm rejects this branch and continues to the next one. Else, given an LP solution $x$, we can assign the values for the agents that we already guessed. Namely for each agent $a\in A_V$, we let 
\[\pi(a)={\kern -5pt} \bigcup_{(e,0)\in \agentPieces_a} {\kern -5pt} \{(e,[0,x^0_e])\}\cup {\kern -5pt} \bigcup_{(e,1)\in \agentPieces_a} {\kern -5pt} \{(e,[1-x^1_e,1])\}.\] 
Every other agent $a$ that we identified via a guess was fully assigned to some particular edge $e$. Moreover, \iflong by Constraint~\ref{eq:piece_length_sum_to_1}, \fi we can split the interval $[x^0_e,1-x^1_e]$ into $n_e$ pieces of length $\delta_e$; note that if $n_e>0$ then $\delta_e$ cannot be equal to $0$. Let $I_a\subseteq [x^0_e,x^1_e]$ be any of the pieces that has not been assigned to another agent yet and let $\pi(a)=(e,I_a)$. 

Finally, we are left with some unassigned agents and some unassigned pieces of the graph, each consisting of a single piece of an edge. Since, $|A_V|+\sum_{e\in E(G)} n_e=|A|$, the number of unsigned pieces equals the number of unassigned agents. Moreover, since at this point we have a concrete partition, for every pair of agent $a$ and piece $(e,I)$ we can in polynomial time check whether $a$ would envy a piece in the partition or not (since this check can be performed without knowing the assignments of the other agents); in the former case we say that $a$ is \emph{compatible} with $(e,I)$, and otherwise we say that they are \emph{incompatible}. 
We can thus create an auxiliary bipartite graph $H=(X\uplus Y, F)$ such that each vertex in $X$ is identified with an unassigned agent, each vertex in $Y$ is identified with an unassigned piece, and there is an edge between an agent $a\in X$ and a piece $(e,I)\in Y$ if and only if they are compatible. We compute a maximum matching $M$ in $H$ \iflong using, e.g., the Hopcroft-Karp algorithm \fi in time at most $\bigoh(|A|^3)$. If $M$ is not a perfect matching, then we reject the branch of our algorithm and try another branch. Else for each unassigned agent $a$, we let \(\pi(a)=M(a)\), where $M(a)$ denotes the piece $(e,I)\in Y$ matched with the agent $a\in X$ by the matching $M$. In this case the algorithm outputs \textsf{Yes} and optionally also the assignment $\pi$ as a witness. If none of the branches lead to a positive outcome, the algorithm outputs \textsf{No}. 

\ifshort
This concludes the description of the algorithm. It now remains to prove correctness and verify the running time. ($\star$)
\fi

\iflong
This concludes the description of the algorithm, and we now proceed to establishing its correctness.

\begin{lemma}\label{lem:No_edges_sufficient_condition}
If the above algorithm outputs an assignment $\pi$, then the assignment $\pi$ is envy-free.
\end{lemma}  
\begin{proof}
	For the ease of presentation, we assume that for $i\in \{0,1\}$ and $e\in E(G)$,  $x^i_e$ and $\delta_e$ are the values of the same-name variables that are output of the LP solver in the branch that leads to the assignment $\pi$.  
	By the construction of the assignment $\pi$ we have that for every agent $a\in A_V$ \[\val_a(\pi(a)) = \sum_{(e,i)\in \agentPieces_a}\val_a(e)\cdot x^i_e,\] and every other agent $a\in A\setminus A_V$ gets a piece that is fully inside some edge $e$ and has length $\delta_e$, hence \[\val_a(\pi_a) = \val_a(e)\cdot\delta_e.\] The fact that for every $(e,i)$, $e\in E(G)$ and $i\in \{0,1\}$, we guessed an agent in $A_V$, that $|A_V|+\sum_{e\in E(G)} n_e=|A|$, and that all constraints of type (\ref{eq:piece_length_sum_to_1}) are satisfied implies that \(\pi(A)\) is indeed a partition of $G$. Constraints of type (\ref{eq:vertex-vertex_envy}) and (\ref{eq:vertex-edge_envy}) guarantee that the agents in $A_V$ do not envy any other agent. Moreover, because for every edge and a pair of agents assigned to that edge in our guessing part we check that Inequalities~(\ref{ineq:edge-edge})~and~(\ref{ineq:edge-vertex}) and because Constraints (\ref{eq:edge-edge_envy}) and (\ref{eq:edge-vertex_envy}) are satisfied, we get that all agents (partially) assigned in our branching phase do not envy any other agent. Finally, all remaining agents are assigned by the matching algorithm, which pairs them only with pieces they value the most and hence none of the agents assigned by the matching algorithm envies any other agent. 
\end{proof}
\begin{lemma}\label{lem:No_edges_necessary_condition}
	If \(\mathcal{I}\) admits an envy-free assignment $\tempAsgn$, then the above algorithm will output an envy-free assignment. 
\end{lemma} 
\begin{proof}
	We will first describe the branch of our algorithm that corresponds to the assignment $\tempAsgn$, explain why we do not reject this branch and show that the LP we obtain in this branch has a solution. Afterwards, we show that any solution to this LP leads to an envy-free assignment.  
	
	Given $\tempAsgn$, we can find the set of agents $A_V$ that get the endpoint pieces of each edge as well as for each agent $a\in A_V$ we know precisely which endpoint the agents gets. For agent $a\in A_V$ we denote by $\agentPieces_a$, in the same way as in the description of the algorithm, the set of pairs $(e,i)$ such that if $i=0$, then $\tempAsgn(a)$ contains piece $(e,[0,c])$, for some $c\in [0,1]$, and if $i=1$, then $\tempAsgn(a)$ contains piece $(e,[d,1])$, for some $d\in [0,1]$. Note that in one of the branches of our algorithm, we get precisely same set of agents $A_V$ with precisely same $\agentPieces_a$ for each agent $a\in A_V$. Moreover, each of agents induces a connected subset of $G$, hence we did not rejected the branch for connectivity reasons. Given this pieces, we can already assign the values for $x^i_e$'s in the LP instance. If $(e,0)\in \agentPieces_a$ and $(e,[0,c])$ is a piece in $\tempAsgn(a)$ such that $c\neq 1$, then we let $x^0_e=c$, otherwise if $c=1$, then we let $x^0_e=\frac{1}{2}$. Similarly, if $(e,1)\in \agentPieces_a$ and $(e,[c,1])$ is a piece in $\tempAsgn(a)$ such that $c\neq 0$, then we let $x^1_e=1-c$, otherwise if $c=0$, then we let $x^1_e=\frac{1}{2}$. Note that \[\val_a(\tempAsgn(a))=\sum_{(e,i)\in \agentPieces_a}\val_a(e)\cdot x^i_e.\] 
	Now, for each edge $e\in E(G)$, we let $n_e$ be the number of agents that, according to \(\tempAsgn\), get a piece fully inside the edge $e$ that does not contain any endpoint of $e$. By Observation~\ref{obs:same_length_pieces_on_edge}, all such pieces have the same length and we let $\delta_e$ denote the length of the interval in such piece. For technical reasons, if $n_e=0$, we let $\delta_e=0$. For any agent $a$ that gets a piece fully inside the edge $e$ we get \[\val_a(\tempAsgn(a)) = \delta_e\cdot \val_a(e).\]
	Observe that in one of our branches we guessed precisely same
        numbers of agents per edge and $|A_V|+\sum_{e\in E(G)}
        n_e=|A|$, so we did not reject this branch. Moreover, since
        \(\tempAsgn\) is envy-free assignment and every agent values
        the whole graph $1$, then each agent gets a non-trivial
        piece. Hence, the given assignment of $x^i_e$'s and
        $\delta_e$'s satisfy all constraints of type
        (\ref{eq:piece_length_positive}). Since \(\tempAsgn(A)\) is a
        partition of $G$, we get also that the constraint
        (\ref{eq:piece_length_sum_to_1}) is satisfied for every edge
        $e$. Since \(\tempAsgn(A)\) is envy-free, we get that all
        constraints of types (\ref{eq:vertex-vertex_envy}) and
        (\ref{eq:vertex-edge_envy}) are satisfied. Now for an ordered
        pair of edges $(e,f)$, if $n_e\ge 1$, then there is an agent
        $a_{(e,f)}$ that minimizes the value of the fraction
        $\frac{\val_{a_{(e,f)}}(e)}{\val_{a_{(e,f)}}(f)}$ among all
        the agents assigned a piece fully inside $e$. Note that this
        agent might not be unique as several agents might have the
        same value of the fraction, however in one of the branches we
        guessed the agents $a_{(e,f)}$ among the ones that minimizes
        the fraction. Since $\tempAsgn$ is envy-free, all constraints
        of type (\ref{eq:edge-edge_envy}) are satisfied. Moreover, in
        a correct branch we also satisfy all inequalities of type
        (\ref{ineq:edge-edge}), so we do not reject this branch
        because of some inequality of type
        (\ref{ineq:edge-edge}). Finally, the $x^i_e$'s define a point
        in $\mathbb{R}^{2k}$ that belong to some cell of
        \(\mathcal{P}=\bigcup_{e\in E(G), a\in A_V, a_1,a_2\in
          A}\{P^{(e,a)}_{a_1,a_2}\}\) over $\mathbb{R}^{2k}$. The
        algorithm branches over all cells of \(\mathcal{P}\) that
        contain a point and hence we branch on some point defined by
        $y^i_e$'s from the same cell as $x^i_e$. It follows that for
        each agent $a\in A_V$, every edge $e\in E(G)$ and every pair
        of agents $a_1,a_2$ we get \[\frac{\sum_{(f,i)\in
              \agentPieces_a}\val_{a_{1}}(f)\cdot
            x^i_f}{\val_{a_{1}}(e)}\ge \frac{\sum_{(f,i)\in
              \agentPieces_a}\val_{a_2}(f)\cdot
            x^i_f}{\val_{a_2}(e)}\] if and only if 
	\[\frac{\sum_{(f,i)\in \agentPieces_a}\val_{a_{1}}(f)\cdot y^i_f}{\val_{a_{1}}(e)}\ge \frac{\sum_{(f,i)\in \agentPieces_a}\val_{a_2}(f)\cdot y^i_f}{\val_{a_2}(e)}.\] 
	Now, for each edge $e\in E(G)$ such that $n_e\ge 1$ and each agent $a\in A_V$, there is an agent $a_{e,a}$ that maximizes the fraction 
	\[\frac{\sum_{(f,i)\in \agentPieces_a}\val_{a_{e,a}}(f)\cdot x^i_f}{\val_{a_{e,a}}(e)}\] 
among all the agents that are assigned a piece fully inside $e$. Since agent $a_{e,a}$ is assigned a piece fully inside $e$ by $\tempAsgn$, we get that 
	\[	\val_{a_{e,a}}(e)\cdot \delta_e \ge \sum_{(f,i)\in \agentPieces_a}\val_{a_{e,a}}(f) \cdot x^i_f.\] 
	
It follows that in the branch where we guessed $a_{e,a}$ correctly all agents that satisfy the inequality (\ref{ineq:edge-vertex}) also satisfy the constraint (\ref{eq:edge-vertex_envy}) and all constraints of type (\ref{eq:edge-vertex_envy}) are satisfied. Moreover, in the correct branch, we do not reject the branch because of inequalities of types (\ref{ineq:edge-edge}) and (\ref{ineq:edge-vertex}). It follows that the branch of our algorithm described above is not rejected and the LP we obtain in branch has a solution.

	Now let us take any solution to the LP obtained in the branch described above and define the assignment \(\pi\) as in our algorithm. It remains 
	to show that the matching instance admits a perfect matching. Since, it is a solution, it is easy to see that under this assignment the agents 
	in $A_V$ do not envy any other agent. Similarly, all the agents that we guessed in our branching do not envy any other piece. Note that all 
	agents assigned the piece so far get the same "kind" of a piece as they get by \(\tempAsgn\). It remains to show that if agent $a\in A$ gets a 
	piece fully inside an edge $e\in E(G)$ according to \(\tempAsgn\), then the auxiliary bipartite graph $H=(X\uplus Y, F)$ that our algorithm 
	constructs contains an edge between $a$ and all pieces fully inside $e$. That is we need to check that assigning the agent $a$ to the piece 
	fully inside $e$ does not cause envy. First note for every $a'\in A_V$, the agent $a$ satisfies the inequality 
	\[\frac{\sum_{(f,i)\in 
	\agentPieces_{a'}}\val_{a_{e,a'}}(f)\cdot y^i_f}{\val_{a_{e,a'}}(e)}\ge \frac{\sum_{(f,i)\in \agentPieces_{a'}}\val_{a}(f)\cdot 
y^i_f}{\val_{a}(e)},\] 

by the choice of the agent $a_{e,a'}$ and hence the LP instance contains the constraint 
\[\val_{a}(e)\cdot \delta_e \ge \sum_{(f,i)\in \agentPieces_{a'}}\val_{a}(f) \cdot x^i_f, \] 
therefore, $a$ does not envy any piece of an agent in $A_V$. Now, for an edge $f$, we know 
that 
\[\frac{\val_{a}(e)}{\val_{a}(f)}\ge \frac{\val_{a_{(e,f)}}(e)}{\val_{a_{(e,f)}}(f)},\] 
by the choice of the agent $a_{(e,f)}$, and that the LP instance contains a constraint 
\[\val_{a_{(e,f)}}(e) \cdot \delta_e \ge \val_{a_{(e,f)}}(f) \cdot \delta_f.\] 
Therefore, 
\[\val_{a}(e) \cdot \delta_e\ge \val_{a}(f) \cdot \delta_f\] 
and $a$ would not envy a piece inside $f$. It follows that an agent $a$ that is assigned a piece fully inside $e$ by \(\tempAsgn\) can be assigned a piece fully inside \(e\) by \(\pi\) as well and \(H\) admits a perfect matching.
\end{proof}
	\begin{lemma}\label{lem:No_edges_running_time}
		The runtime of the above algorithm is $|A|^{\bigoh(k^2)}$.
	\end{lemma} 
	\begin{proof}
		All operations such as adding a constraints in LP instance, checking the consistency of the instance, solving LP, creating the auxiliary graph $H$ and finding maximum size matching in $H$ can be done in polynomial time. Moreover, we can enumerate points in all cells of \(\mathcal{P}\) over $\mathbb{R}^{2k}$ in \(|A|^{\bigoh(k)}\) time. Hence, to get running time $|A|^{\bigoh(k^2)}$, it suffices to upper bound the number of branches by $|A|^{\bigoh(k^2)}$. We first branch for each endpoint of an edge for the agent that gets this endpoints. This is $|A|^{2k}$ many branches. For each branch we then branch for each edge on the number of agents that are assigned a piece fully inside the edge. This is $(|A|+1)^k$ many branches. Afterwards, we branch for each pair of edges $e,f$ on the agent $a_{(e,f)}$. This is $|A|^{k^2}$ many branches. In each branch we run the algorithm from Proposition~\ref{prop:sample_points} and branch on $|A|^{\bigoh(k)}$ many points in $\mathbb{R}^{2k}$. This is $|A|^{\bigoh(k)}$ many branches. Finally in each branch, we branch on agents $a_{e,a}$, this is at most $|A|^{2k^2}$ many branches (as $|A_V|\le 2k$). Therefore, the total number of branches is \[|A|^{2k}\cdot (|A|+1)^k \cdot |A|^{k^2}\cdot  |A|^{\bigoh(k)}\cdot|A|^{2k^2} = |A|^{\bigoh(k^2)}. \qedhere \]
	\end{proof}
	The proof of Theorem~\ref{thm:No_edges_main} follows immediately from Lemmas~\ref{lem:No_edges_sufficient_condition}-\ref{lem:No_edges_running_time}.
	\fi

\section{Concluding Remarks}
\label{sec:discussion}

Our results provide a significantly improved understanding of the classical complexity of envy-free graph cutting. One direction that may be of interest for future work is to analyze the complexity of this problem through the more refined \emph{parameterized complexity} paradigm. Indeed, in that setting the algorithmic results presented here can be viewed as \XP\ algorithms.
An immediate question in this context is whether these results can be strengthened to show fixed-parameter tractability. Most prominently, is there a fixed-parameter algorithm for \problemshort\ parameterized by the number of edges?

For the special case where the underlying graph is a path, the complexity of the problem is an even more intriguing question.  As \problemshort\ on a path is a special case of \EF\  \contigcake, we know that  it always admits a solution and hence the decision version of the problem cannot be \W[1]-hard ({\em for any parameterization}). 
If the problem of computing an envy-free solution does not admit a fixed-parameter algorithm, would showing this require a variation of the \W-hierarchy tailored specifically to TFNP problems\iflong~\cite{MegiddoP91}\fi?

\newpage

\section*{Acknowledgements}
Ganian and Hamm acknowledge support from the Austrian Science Fund (FWF, Projects P31336 and Y1329). Hamm also acknowledges support from FWF (Project W1255-N23). Ordyniak acknowledges support by the EPSRC (EP/V00252X/1).

\bibliographystyle{named}
\bibliography{references-short}

\end{document}